\documentclass{IEEEtran}
\usepackage{amsmath}
\usepackage{nccmath}
\usepackage[utf8]{inputenc} %unicode support
\usepackage{epsfig,scalefnt,multirow}
\usepackage{url}
\usepackage{ifthen}
\usepackage{mathtools}
\usepackage{cite}
\usepackage{graphicx}
\usepackage{amssymb}
\usepackage{tabularx}
\usepackage{amsmath}
\usepackage{epstopdf}
\usepackage{epsf}
\usepackage{algorithm}
\usepackage{algpseudocode}
\usepackage{algpascal}
\usepackage{cases}
\usepackage{caption}
\usepackage[labelformat=simple]{subcaption}

\usepackage{stfloats}
\usepackage{float}
\usepackage{xcolor}
\usepackage{tabularx}% http://ctan.org/pkg/tabularx

\usepackage{epsfig,amsmath,amssymb,epsf,amsthm,scalefnt,multirow}
\usepackage{xcolor}
\usepackage{float}
\usepackage{cite}
\usepackage{psfrag}
\usepackage{gensymb}
\usepackage{cleveref}
\usepackage{mathrsfs}
\usepackage{mathtools}
\usepackage{algpseudocode}
\usepackage{algorithm}
\usepackage{enumerate}

  \def\cC{{\mathcal{C}}}

 \def\cN{{\mathcal{N}}}

\def\diag{\mathop{\mathrm{diag}}}

\def\bSigma{{\pmb{\Sigma}}} 
\def\bPhi{{\pmb{\Phi}}} \def\bphi{{\pmb{\phi}}}

\def\b0{{\pmb{0}}} 

\def\ba{{\mathbf{a}}}   
\def\bee{{\mathbf{e}}} \def\bff{{\mathbf{f}}} \def\bg{{\mathbf{g}}} \def\bh{{\mathbf{h}}}
   
\def\bm{{\mathbf{m}}}   
\def\bq{{\mathbf{q}}} \def\br{{\mathbf{r}}} \def\bs{{\mathbf{s}}} 
\def\bu{{\mathbf{u}}} \def\bv{{\mathbf{v}}}  
\def\by{{\mathbf{y}}}   

\def\bA{{\mathbf{A}}}   
   
\def\bI{{\mathbf{I}}}   
   
   \def\bT{{\mathbf{T}}}

\DeclarePairedDelimiter\norm{\lVert}{\rVert}

\newtheorem{lemma}{Lemma}
\newtheorem*{lemma*}{Lemma}

\begin{document}
	%
	% paper title
	% Titles are generally capitalized except for words such as a, an, and, as,
	% at, but, by, for, in, nor, of, on, or, the, to and up, which are usually
	% not capitalized unless they are the first or last word of the title.
	% Linebreaks \\ can be used within to get better formatting as desired.
	% Do not put math or special symbols in the title.

	\title{IRS-Aided Energy Efficient UAV Communication}

	%
	%
	% author names and IEEE memberships
	% note positions of commas and nonbreaking spaces ( ~ ) LaTeX will not break
	% a structure at a ~ so this keeps an author's name from being broken across
	% two lines.
	% use \thanks{} to gain access to the first footnote area
	% a separate \thanks must be used for each paragraph as LaTeX2e's \thanks
	% was not built to handle multiple paragraphs
	%
	
	\author{Hyesang Cho,~\IEEEmembership{Student Member,~IEEE}	and~Junil Choi,~\IEEEmembership{Senior Member,~IEEE}% <-this % stops a space
	\thanks{The authors are with the School of Electrical Engineering, Korea
		Advanced Institute of Science and Technology, Daejeon 34141, South Korea
		(e-mail: \{nanjohn96, junil\}@kaist.ac.kr).

	}% <-this % stops a space
		%\thanks{Junil Choi is the corresponding author.}
}

	\maketitle
	
	% As a general rule, do not put math, special symbols or citations
	% in the abstract or keywords.

\begin{abstract} \label{sec:abs} % abstract
Unmanned aerial vehicles (UAVs) have steadily gained attention to overcome the harsh propagation loss and blockage issue of millimeter-wave communication. However, UAV communication systems suffer from energy consumption, which limits the flying time of UAVs. In this paper, we propose several UAV energy consumption minimization techniques through the aid of multiple intelligent reflecting surfaces (IRSs). In specific, we introduce a tractable model to effectively capture the characteristics of multiple IRSs and multiple user equipments (UEs). Then, we derive a closed form expression for the UE achievable rate, resulting in tractable optimization problems. Accordingly, we effectively solve the optimization problems by adopting the successive convex approximation technique. To compensate for the high complexity of the optimization problems, we propose a low complexity algorithm that has marginal performance loss. In the numerical results, we show that the proposed algorithms can save UAV energy consumption significantly compared to the benchmark with no IRSs, justifying that exploiting the IRSs is indeed favorable to UAV energy consumption minimization.
\end{abstract}

%%%%%%%%%%%%%%%%%%%%%%%%%%%%%%%%%%%%%%%%%%%%%%
%%                                          %%
%% The keywords begin here                  %%
%%                                          %%
%% Put each keyword in separate \kwd{}.     %%
%%                                          %%
%%%%%%%%%%%%%%%%%%%%%%%%%%%%%%%%%%%%%%%%%%%%%%


% MSC classifications codes, if any
%\begin{keyword}[class=AMS]
%\kwd[Primary ]{}
%\kwd{}
%\kwd[; secondary ]{}
%\end{keyword}

%\end{fmbox}% uncomment this for twcolumn layout

%%%%%%%%%%%%%%%%%%%%%%%%%%%%%%%%%%%%%%%%%%%%%%
%%                                          %%
%% The Main Body begins here                %%
%%                                          %%
%% Please refer to the instructions for     %%
%% authors on:                              %%
%% http://www.biomedcentral.com/info/authors%%
%% and include the section headings         %%
%% accordingly for your article type.       %%
%%                                          %%
%% See the Results and Discussion section   %%
%% for details on how to create sub-sections%%
%%                                          %%
%% use \cite{...} to cite references        %%
%%  \cite{koon} and                         %%
%%  \cite{oreg,khar,zvai,xjon,schn,pond}    %%
%%  \nocite{smith,marg,hunn,advi,koha,mouse}%%
%%                                          %%
%%%%%%%%%%%%%%%%%%%%%%%%%%%%%%%%%%%%%%%%%%%%%%

%%%%%%%%%%%%%%%%%%%%%%%%% start of article main body
% <put your article body there>
\section{Introduction}\label{sec:intro}

Current wireless communication systems are setting foot on millimeter-wave (mmWave) communication via fifth-generation technology \cite{mmWave1, mmWave2, mmWave3}. By exploiting the abundant and vacant frequency bands, mmWave communication has the potential to serve the increasing demand of data. Nonetheless, the high frequency region has its disadvantages, e.g., due to the characteristics of high frequency signals, the signals suffer from extreme propagation loss and are vulnerable to blockage \cite{mmPotential1, mmPotential2}. 

%To relieve these issues, there have been many technological approaches such as intelligent reflecting surfaces (IRSs) and unmanned aerial vehicles (UAVs). 

To overcome the blockage issue, intelligent reflecting surfaces (IRSs) have been proposed as a favorable candidate \cite{IRS1,IRS2,IRS3}. An IRS is a planar surface consisting of multiple passive elements, where each element has the ability to independently shift the phase of the electromagnetic waves impinging on itself \cite{IRSwork1, IRSwork2,IRSwork3,IRSwork4}. By thoroughly adjusting the phase controller, the IRS can enhance the signal gain for a desired location, resulting in a virtual line-of-sight (LoS) path. Therefore, the IRS can function as a controllable relay with minimal power consumption.

Another effective approach to overcome the blockage issue is to exploit unmanned aerial vehicles (UAVs) \cite{UAV1, UAV2,UAV3,UAV4}. A UAV is a flying platform that freely hovers in the sky, which can be used for communication by using it as a mobile base station (BS) or a relay \cite{UAVBS1,UAVBS2,UAVrelay1,UAVrelay2}. Unlike conventional ground nodes, the UAV can enjoy favorable channel conditions and avoid blockage via their high altitude.

Even with their potential, UAV communication systems have their limitations such as energy consumption \cite{UAVenergy1}. While traditional devices replenish energy from a dedicated source, it is improbable for the UAV to replenish energy due to its wireless feature. Also, due to its hovering characteristic, the UAV has additional power consumption, which is usually orders of magnitude larger than communication power \cite{UAVenergy2}. To address this issue, there have been works to maximize the energy efficiency of UAV communication systems. Especially, \cite{Bench1} minimized the energy consumption of a rotary-wing UAV communication system with a specific data constraint through jointly optimizing the UAV trajectory, velocity, and communication time.

Recently, there have been attempts to combine UAV communication systems with IRSs, which can be divided into two large categories \cite{UAVIRS1,UAVIRS2,UAVIRS3,UAVIRS4}. In the first category, the UAV itself is equipped with the IRS. In particular, \cite{UAVonIRS1} analyzed the outage probability, ergodic capacity, and energy efficiency  of wireless communication systems supported by the IRS equipped UAVs. In the second category, the UAV communication system is assisted through terrestrial IRSs. In \cite{UAVgroundIRS1}, it has been shown that exploiting a terrestrial IRS is beneficial for UAV energy consumption minimization. However, \cite{UAVgroundIRS1} assumed the LoS channel model and performed the task with only a single IRS. Also, \cite{UAVgroundIRS2} maximized the average reception power in a given time constraint using multiple IRSs by optimizing the IRS phase shifts, transmit beamformer, and UAV trajectory. While \cite{UAVgroundIRS2} discussed a multiple IRS scenario, it did not reveal the tradeoff between the reception power and UAV flight power, e.g., the UAV may consume excess flight power to maximize the reception power. Accordingly, this tempts us to explore UAV energy consumption minimization in a multiple terrestrial IRS scenario.

In this paper, we propose several UAV energy consumption minimization algorithms with a specific data constraint for various scenarios including the most general multiple IRS multiple UE (MIMU) case. The main insight is that by increasing the received data rate at the UEs through the IRSs, we can reduce the flight time, thus, reduce the UAV energy consumption. In specific, we first expand the probabilistic LoS model to capture the advantage of deploying IRSs on tall buildings \cite{UAVProbLOS}. To the best of our knowledge, this is the first work to consider the popular probabilistic LoS channel model for UAV scenarios in IRS assisted environments. By utilizing the new system model, we derive the optimal IRS phase shifts and develop a closed form of the UE achievable rate. We then formulate the optimization problems that jointly optimize the trajectory, flight speed, and communication time of the UAV. Afterwards, we transform the optimization problems into tractable forms exploiting slack variables and effectively solve them through the well known successive convex approximation (SCA) technique \cite{UAVIRS3}. Finally, we propose a low complexity algorithm that tries to follow the behavior of a specific optimization problem solution, resulting in comparable performance with respect to the jointly optimized results.

The rest of paper is organized as follows. Section \ref{sec:model} describes the channel model of the MIMU scenario and the power consumption model of the UAV. In Section \ref{sec: single IRS}, we derive and formulate the UAV energy consumption minimization problem for the single IRS case. In Section \ref{sec: multi IRS}, we generalize the system to the multiple IRS case, and introduce several algorithms to effectively solve the joint optimization problems. Section \ref{sec: simul} shows the simulation results of the proposed techniques with a benchmark scheme without any IRSs. Finally, Section \ref{sec:concl} concludes the paper. 

\textbf{Notation:} Lower and upper boldface letters represent column vectors and matrices. $\bA^{*}$  and $\bA^{\mathrm{H}}$ denotes the conjugate, and conjugate transpose of the matrix $\bA$. $\diag(\ba)$ returns the diagonal matrix with $\ba$ on its diagonal. $\left\{ \ba \right\}^b$ represents the set of length $b$ vectors with each element from $\ba$. ${\mathbb{C}}^{m \times n}$ and ${\mathbb{R}}^{m \times n}$ represent the set of all $m \times n$ complex and real matrices. $|{\cdot}|$ denotes the amplitude of the scalar, and $\norm{\cdot}$ represents the $\ell_2$-norm of the vector.  $\mathcal{O}$ denotes the Big-O notation. $\boldsymbol{0}_m$ is used for the $m\times1$ all zero vector, and $\bI_m$ denotes the $m \times m$ identity matrix. $\cC \cN(\bm,\bSigma)$ denotes the circularly symmetric complex Gaussian distribution with mean $\bm$ and variance $\bSigma$.

 \begin{figure}[t] 
	\centering
	\includegraphics[width=1 \columnwidth]{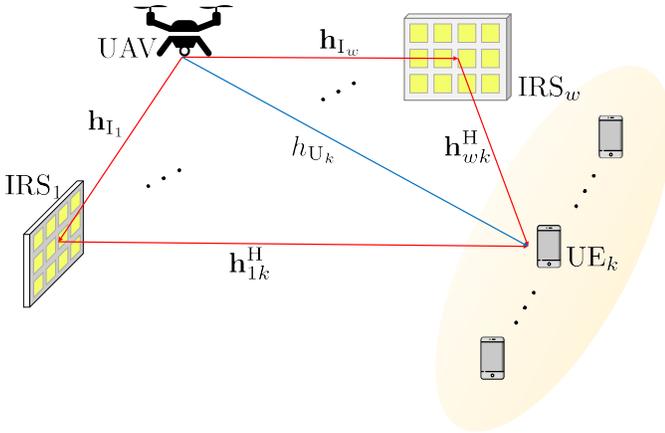}
	%   % where an .eps filename suffix will be assumed under latex,
	%   % and a .pdf suffix will be assumed for pdflatex
	\caption{Single UAV downlink system with $K$ UEs and $W$~IRSs.} 
	\label{fig:System model}
\end{figure}

\section{System Model}\label{sec:model}
In this paper, we consider a wireless communication system where a rotary-wing UAV supports multiple UEs with the aid of multiple IRSs as in Fig. \ref{fig:System model}. We assume a single antenna UAV, $K$ single antenna UEs, and $W$ IRSs, with the $w$-th IRS having $M_w$ IRS elements. The UAV freely hovers around the horizontal plane, where the horizontal location at time $t$ is denoted as $\bq (t) \in \mathbb{R}^{2\times 1} $, and the altitude is fixed as $H_\text{A} \in \mathbb{R}$. Since the IRSs are typically located to gain a better LoS condition, we assume that the IRSs are installed on the outer walls of tall buildings, thus, having more height than the UEs. We assume that the $w$-th IRS is fixed with the horizontal location $\bq_{\text{I}_w} \in \mathbb{R}^{2\times 1}$ and height $H_{\text{I}_w} \in \mathbb{R}$. Also, we assume that the UEs are located without mobility with the horizontal location and height of the $k$-th UE denoted as $\bq_{\text{U}_k} \in \mathbb{R}^{2\times 1}$ and $H_{\text{U}_k} \in \mathbb{R}$, respectively.

\subsection{Channel Model}
The channel between any two nodes in the system, i.e., UAV, IRSs, and UEs, are denoted as
	\begin{align}\label{eq:channel}
	\tilde{\bh}_{\text{I}_w} (t) &= \sqrt{\frac{\beta_0}{d^{\alpha_{\text{I}_w}}_{\text{I}_w} (t)}} \tilde{\bg}_{\text{I}_w} (t), \\
	\tilde{\bh}_{wk} (t) &= \sqrt{\frac{\beta_0}{d^{\alpha_{wk}}_{wk} }} \tilde{\bg}_{wk} (t), \\
	\tilde{h}_{\text{U}_k} (t) &= \sqrt{\frac{\beta_0}{d^{\alpha_{\text{U}_k}}_{\text{U}_k} (t)}} \tilde{g}_{\text{U}_k} (t),
	\end{align}
where $\tilde{\bh}_{\text{I}_w} (t) \in \mathbb{C}^{M_w \times 1}, \tilde{\bh}_{wk}(t) \in \mathbb{C}^{ M_w \times 1}$, and $\tilde{h}_{\text{U}_k} (t) \in \mathbb{C}$ denote the UAV to the $w$-th IRS channel, the $w$-th IRS to the $k$-th UE channel, and the UAV to the $k$-th UE channel, respectively. We split the channel into the large scale fading factors and small scale fading channel, where for the large scale fading factors, the distances between the UAV to the $w$-th IRS, the $w$-th IRS to the $k$-th UE, and the UAV to the $k$-th UE are denoted as $d_{\text{I}_w}(t), d_{wk},$ and $d_{\text{U}_k}(t)$, respectively, where we observe that $d_{wk}$ is constant since the locations of UEs and IRSs are fixed. The pathloss exponent is denoted as $\alpha_i, \ i \in \left\{\text{I}_w, wk,\text{U}_k\right\}$, and $\beta_0$ is the pathloss at a reference distance. The small scale fading channels are denoted as
\begin{align}\label{eq:small scale channel}
\tilde{\bg}_{i} (t) &= \sqrt{\frac{\kappa_i}{\kappa_i +1}} \bee_i (t) + \sqrt{\frac{1}{\kappa_i +1}} \bar{\bg}_i (t), \ i \in  \left\{ \text{I}_w, wk\right\}, \\
\tilde{g}_{\text{U}_k} (t) &= \sqrt{\frac{\kappa_{\text{U}_k}}{\kappa_{\text{U}_k} +1}} \exp(j\psi_{\text{U}_k}) (t) + \sqrt{\frac{1}{\kappa_{\text{U}_k} +1}} \bar{g}_{\text{U}_k} (t),
\end{align}
where $\tilde{\bg}_{\text{I}_w} (t), \tilde{\bg}_{wk}(t)$, and $\tilde{g}_{\text{U}_k} (t)$ denote the small scale fading channels between the UAV to the $w$-th IRS, the $w$-th IRS to the $k$-th UE, and the UAV to the $k$-th UE, respectively. The vectors $\bee_{\text{I}_w} = \left[ \exp (j \psi_{\text{I}_{w,1}}), ... , \exp (j \psi_{\text{I}_{w,{M_w}}}) \right]^\text{T}$ and $\bee_{wk} = \left[ \exp (j \psi_{w,1k}), ... , \exp (j \psi_{w,{M_w}k}) \right]^\text{T}$ denote the LoS path phases with $\psi_{\text{I}_{w,m}}$ and $\psi_{w,mk}$ representing the phases of the channels between the UAV to the $m_w$-th IRS element of the $w$-th IRS and the $m_w$-th IRS element of the $w$-th IRS to the $k$-th UE, respectively. The LoS path phase of the channel between the UAV to the $k$-th UE is denoted as $\psi_{\text{U}_k}$.The small scale channels follow the independent Rician distribution with unit power and the $K$-factor as $\kappa_i, \ i \in \left\{\text{I}_w, wk,\text{U}_k\right\}$, and with $\bar{\bg}_{\text{I}_w}(t) \sim \cC\cN \left(\b0_{M_w},\bI_{M_w}\right), \bar{\bg}_{wk}(t) \sim \cC\cN \left(\b0_{M_w},\bI_{M_w}\right)$, and $\bar{g}_{\text{U}_k}(t) \sim \cC\cN \left(0,1\right)$.

We assume that the LoS paths between the IRSs and UEs always exist. For the channels associated with the UAV, we adopt the well-known probabilistic LoS model \cite{UAVProbLOS}, which models the LoS path existence probability as a variable dependent on the elevation angle between a node and a UAV. The LoS path existence probability is given as \cite{UAVProbLOS}
\begin{align} \label{eq: los prob}
p_i(t) = \frac{1}{1+a_i \exp \left( -b_i \left( \theta_i(t) -a_i \right) \right) }, \ i \in \left\{ \text{I}_w, \text{U}_k \right\},
\end{align}
where $a_i$ and $b_i$ are the design parameters dependent on the environment, and $\theta_i(t)$ is the elevation angle, which can be expressed as $\theta_i(t) = \frac{180}{\pi} \sin^{-1} \left( \frac{H_\text{A}- H_i}{d_i(t)}\right)$. Note that, due to the height of the installed IRSs, we consider the design parameters $a_i$ and $ b_i$ as variables. Defined in \cite{UAVProbLOS}, the design parameters are dependent on the average building height and building density with respect to the ground. Therefore, by taking the height of the installed IRSs on buildings as the new effective ground, the average building height and density which the IRSs see will decrease, resulting in a relatively rural environment with respect to the original environment. This new consideration can effectively capture the fact that highly located IRSs are favorable in obtaining the LoS path. Note that, the design parameters $a_i$ and $b_i$ can adapt to each IRS or UE, thus, this model can effectively consider the different heights of the IRSs or even consider the advantages of highly located UEs. Accordingly, the effective channel between the nodes can be expressed as
\begin{align}
\bh_{\text{I}_w} (t)&=
\begin{cases}
\tilde{\bh}_{\text{I}_w}(t), \quad  \  \text{Probability} \ p_{\text{I}_w}(t), \\
\nu_{\text{I}_w} \tilde{\bh}_{\text{I}_w}(t), \quad \text{Probability} \ 1-p_{\text{I}_w}(t),
\end{cases} \\
h_{\text{U}_k} (t)&=
\begin{cases}
\tilde{h}_{\text{U}_k}(t), \quad  \  \text{Probability} \ p_{\text{U}_k}(t), \\
\nu_{\text{U}_k} \tilde{h}_{\text{U}_k}(t), \quad \text{Probability} \ 1-p_{\text{U}_k}(t),
\end{cases}
\end{align}
where $\nu_i < 1, \ i \in \left\{ \text{I}_w, \text{U}_k\right\},$ is the additional attenuation factor due to the loss of the LoS path. From hereafter, we take $\nu_i = 0$ for equational conciseness, which will be explained in Section \ref{sec: single IRS}. Note that, the proposed methods can still be applied when $\nu_i \ne 0$. In conclusion, the overall channel for the $k$-th UE is given as
\begin{align}
\bar{h}_k (t) = \sum^W_{w=1}   \left( \bh_{wk}^{\text{H}}(t) \bPhi_w (t)\bh_{\text{I}_w}(t)\right) + h_{\text{U}_k}(t),
\end{align}
where $\bPhi_w $ is the phase control matrix of the $w$-th IRS. The phase control matrix can be represented as $\bPhi_w = \text{diag} \left( \bphi_w \right)$, with $\bphi_w = \left[ \exp(j \phi_{w,1}), ..., \exp ( j \phi_{w,{M_w}})\right] ^ \text{T}$, thus, the $m$-th IRS element of the $w$-th IRS will shift the phase of the impinging signal with the factor $\phi_{w,m}$.

We adopt the time-division multiple access (TDMA) technique at the UAV to serve the $K$ UEs to efficiently prevent inter-user interference. The achievable rate for the $k$-th UE can be denoted as 
\begin{align}
\bar{R}_k (t) = B \log_2 \left( 1+ \frac{P_c \vert \bar{h}_k(t) \vert^2}{B \sigma^2}\right),
\end{align}
where $B$ is the bandwidth, $P_c$ is the transmit power from the UAV, and $\sigma^2$ denotes the noise spectral density.

\subsection{UAV Power Consumption Model}
Since we assumed to use the rotary-wing UAV, we adopt the rotary-wing UAV flight power consumption model from \cite{Bench1}, which is given as
\begin{align}
P(V) &= P_0 \left( 1+ \frac{3V^2}{U_{\text{tip}}^2}\right) + P_i \left( \sqrt{1+ \frac{V^4}{4v_0^4}} - \frac{V^2}{2v_0^2}\right)^{1/2} \notag \\ 
&+ \frac{1}{2} d_0\rho sA_\text{p}V^3,
\end{align}
where $V$ is the UAV speed and the other variables are listed in Table \ref{table: UAVpower}. The total energy consumption of the UAV can be expressed as
\begin{align}
E_\text{tot} = \tau P_c+	\int_{0}^{T} P\left( V(t)\right) dt,
\end{align}
where $\tau$ is the communication time and $T$ is the total flight time of the UAV. We observe that the power consumption is dependent on the UAV speed, and that the model itself is quite complicated. This leads us to believe that the UAV speed should also be jointly optimized to effectively minimize the UAV power consumption. Also, we can confirm that while the power consumption for communication usually has the magnitude of $1$ W or smaller, the power consumption from UAV hovering has the magnitude of $100$ W. Thus, it is obvious that we must consider the UAV flight power to successively achieve energy efficient UAV communication systems.

\begin{table}[t]
\begin{center}
	\caption{UAV flight power consumption variables.}
	\label{table: UAVpower}
	\begin{tabular}{ |l|l|l|}
		\hline
		Variable & Description & Value \\
		\hline
		$P_0$   & Blade profile power & $79.86$ \\
		$P_i$& Induced power & $88.63$\\
		$U_\text{tip}$ & Rotor blade tip speed   & $120$ \\
		$v_0$ & Rotor induced velocity  & $4.03$ \\
		$d_0$&  Fuselage drage ratio &  $0.6$ \\
		$\rho$&  Air density  & $1.225$ \\
		$s$&  Motor solidity & $0.05$ \\
		$A_\text{p}$ &  Rotor disc area & $0.503$ \\
		\hline
	\end{tabular}
\end{center}
\end{table}

\section{Single IRS Optimization}\label{sec: single IRS}
In this section, we define and solve the optimization problems to minimize the UAV energy consumption in a simple single IRS case. Specifically, we first assume the most elementary single IRS single UE (SISU) case, where we derive the closed form expression for the IRS phase shifts and the achievable rate. Thereafter, we define the joint optimization problem and develop it into a tractable form. Next, we expand the scenario to the single IRS multiple UE (SIMU) case, which can be easily extended from the SISU case. The last expansion to the MIMU case will be held on Section \ref{sec: multi IRS}.

\subsection{SISU Scenario Achievable Rate}
In the SISU case, we omit the indexes $w$ and $k$ for brevity. To minimize the UAV energy consumption while sufficing a specific data constraint, it is trivial to operate the IRS phase shifts to maximize the achievable rate. Since the UAV and UE are both single antenna devices, the optimal phase shifts for the IRS are given in a simple closed form. Note that,  we need the instantaneous channel state information to derive the optimal phase shifts. To focus on the achievable theoretical performance, we assume to have perfect channel state information (CSI) throughout this paper. The reflection coefficient of the $m$-th IRS element is then given as
\begin{align} \label{eq: phase shift}
\phi_m (t) = \angle \left( h_\text{U} (t) \right) - \angle \left( h_{\text{I},m} (t) \right) + \angle \left( h_m (t)\right),
\end{align}
where $h_{\text{I},m} (t)$ and $ h_m (t)$ denote the $m$-th element of the channel vectors $\bh_\text{I}(t)$ and $\bh(t)$, respectively. This will coherently combine the UAV-UE channel with the UAV-IRS-UE channel, optimally maximizing the achievable rate. 

To gain tractability, we adopt two frequently used assumptions as in \cite{Bench1}. First, to capture the randomness of the channel, we use the expected achievable rate instead of the instantaneous achievable rate. Note that, this might seem as if perfect CSI is unnecessary. However, to obtain the expected achievable rate, we need the IRS to constantly adapt to the instantaneous channels for successful coherent combination, thus, the assumption of perfect CSI is necessary. Second, due to the involved form of \eqref{eq: los prob}, we assume that the LoS path existence probabilities are constant with respect to time, e.g., fix the probabilities with the average elevation angles. Nevertheless, the simplified model still captures the advantage of the highly located IRSs.

To effectively express the probabilistic LoS model, we define a state variable as 
\begin{align}
\bs = \left[ s_\text{U}, s_\text{I} \right]^{\text{T}}, \ s_i \in \left\{ 0, 1\right\},
\end{align}
for $i \in \left\{ \text{U}, \text{I}\right\}$, where $s_\text{U}$ and $s_\text{I}$ denote the states of the UAV-UE path and UAV-IRS path, respectively, and $s_i=1$ represents the existence of the selected path. With this variable, we can upper bound the expected achievable rate using the Jensen's inequality as
\begin{align}
&\mathbb{E} \left[ \bar{R}(t)\right]\notag \\
 &= \sum_{\bs \in \left\{ 0,1\right\}^2} \prod_i p_i^{s_i} \left(1-p_i\right)^{1-s_i} \mathbb{E} \left[ B \log_2 \left( 1+ \frac{P_c \vert \bar{h}_{\bs} (t) \vert ^2}{B\sigma^2}\right)\right]\\
&\leq \sum_{\bs \in \left\{ 0,1\right\}^2} \prod_i p_i^{s_i} \left(1-p_i\right)^{1-s_i} B\log_2 \left( 1+ \frac{P_c\mathbb{E} \left[ \vert \bar{h}_{\bs}(t) \vert ^2 \right]}{B\sigma^2}\right) \label{eq: jensen}\\
&=\hat{R}(t), \  i \in \left\{ \text{U}, \text{I}\right\}, \notag
\end{align}
where $\bar{h}_\bs (t)$ is the overall channel between the UAV and the UE. Due to the channel states, the overall channel is now a variable dependent on $\bs$. We observe that the closed form of \eqref{eq: jensen} can be derived by computing $\mathbb{E} \left[ \vert \bar{h}_{\bs}(t) \vert ^2 \right]$. However,  $\mathbb{E} \left[ \vert \bar{h}_{\bs} (t) \vert ^2 \right]$ is difficult to derive due to the joint consideration of the UAV-UE channel and the UAV-IRS-UE channel.

%\begin{figure*}
%	\centering
%	\begin{subfigure}[b]{0.45\textwidth}
%		\centering
%		\includegraphics[width=\textwidth]{Ratecompare1.eps}
%		\caption{Single IRS scenario}
%		\label{fig:Ratecompare1}
%	\end{subfigure}
%	\begin{subfigure}[b]{0.45\textwidth}
%		\centering
%		\includegraphics[width=\textwidth]{Ratecompare2.eps}
%		\caption{Multiple IRS scenario}
%		\label{fig:Ratecompare2}
%	\end{subfigure}
%	\hfill
%	\caption{Achievable rate comparison between the approximated analytical values and the simulated results, where \ref{fig:Ratecompare2} has three IRSs. For both cases, the UE is located at $250$ m and the IRSs are located within a $5$ m proximity of the UE.}
%	\label{fig:Ratecompare}
%\end{figure*}

\begin{figure}[t] 
	\centering
	\includegraphics[width=1 \columnwidth]{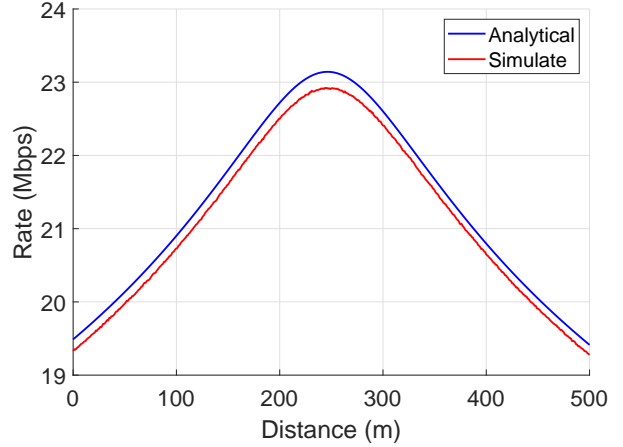}
	%   % where an .eps filename suffix will be assumed under latex,
	%   % and a .pdf suffix will be assumed for pdflatex
	\caption{Achievable rate comparison between the approximated analytical values and the simulated results for the SISU case. The UE is located at $250$ m and the IRS is located at $245$ m.} 
	\label{fig:Rate compare}
\end{figure}

To compute the expected channel gain, we first assume that the IRS phases are always coherently combining the UAV-UE channel and the UAV-IRS-UE channel. In result, the overall channel norm value is given as
\begin{align} \label{eq: channel const}
	\lvert \bar{h}_\bs (t) \rvert =s_\text{U} \lvert  h_\text{U} (t)\rvert +s_\text{I} \sum^M_{m=1} \lvert h_{\text{I},m}(t) \rvert \lvert h_m (t)\rvert.
\end{align}
By defining a new variable $C(t) = \sum^M_{m=1}  \lvert h_{\text{I},m}(t) \rvert \lvert h_m (t)\rvert $, we can expand the expected channel gain as
\begin{align} \label{eq: channel gain}
\mathbb{E} \left[ \vert \bar{h}_{\bs} (t) \vert ^2 \right] &= s_\text{U}^2 \mathbb{E} \left[ \lvert h_\text{U}(t) \rvert^2 \right] \notag \\
&+ 2 s_\text{U} s_\text{I} \mathbb{E} \left[ \lvert h_\text{U} (t)\rvert \right] \mathbb{E} \left[ C(t) \right] + s_\text{I}^2 \mathbb{E} \left[ C^2(t) \right],
\end{align}
where $\lvert h_\text{U} (t)\rvert C(t)$ can be separated due to independence. Note that $C(t)$ is the sum of $M$ i.i.d. variables. Considering that the number of IRS elements, i.e., $M$, is predicted to be installed in large quantities, we employ the central limit theorem and assume that $C(t)$ follows the Gaussian distribution. Overall, \eqref{eq: jensen} can be shown as
\begin{align}  \label{eq: single IRS rate}
\hat{R}(t) &= \sum_{\bs \in \left\{ 0,1\right\}^2} \prod_i p_i^{s_i} \left(1-p_i\right)^{1-s_i} \log_2 \left( 1+ \hat{\gamma}_\bs (t)\right), 
\end{align}
with the overall signal-to-noise ratio (SNR) in the closed form expression as
\begin{align} 
\hat{\gamma}_\bs(t) &= s_\text{U} \gamma_\text{U}^2 d_\text{U}^{-\alpha_\text{U}} (t) + 2s_\text{U} s_\text{I} \gamma_\text{U}' \gamma_\text{I}' d_\text{U}^{-\alpha_\text{U}/2} (t)d_\text{I}^{-\alpha_\text{I}/2}(t) \notag \\ &+ s_\text{I} \gamma_\text{I}^2 d_\text{I}^{-\alpha_\text{I}}(t),
\end{align}
where $s_i^2$ is denoted as $s_i$ since $s_i^2 = s_i$, and the time independent variables $\gamma_i$ and  $\gamma_i'$ are defined as
\begin{align}
\gamma_\text{U} &= \sqrt{\frac{\beta_0 P_c}{B\sigma^2}}, \label{eq: gam1} \\
\gamma_\text{U}' &= \sqrt{\frac{\beta_0 P_c}{B\sigma^2}} \mu \left(\kappa_\text{U}\right), \\
\gamma_\text{I} &= \sqrt{\frac{\beta_0^2 P_c}{d^\alpha B\sigma^2}} \sqrt{\mu_\text{I}^2+1}, \\
\gamma_\text{I}' &= \sqrt{\frac{\beta_0^2 P_c}{d^\alpha B\sigma^2}} \mu_\text{I}, \label{eq: gam4}
\end{align}
where $\mu\left(\kappa_i\right)$ denotes the mean of the Rician distribution with the $K$-factor as $\kappa_i$, and $\mu_\text{I} = M \mu \left( \kappa_\text{I}\right) \mu \left( \kappa\right)$. Note that in \eqref{eq: channel const}, the channel terms without LoS paths are zero due to the assumption $\nu_i = 0$, where the case $\nu_i \ne 0$ gives only constant differences in \eqref{eq: gam1}-\eqref{eq: gam4}.

In Fig. \ref{fig:Rate compare}, we plot the approximated achievable rate and the simulated achievable rate for the SISU case with the same parameters in Section \ref{sec: simul}. We observe that the approximated achievable rate is tightly bound with the simulated achievable rate. Thus, hereinafter, we consider $\hat{R}(t)$ as the achievable rate of the UE.

\subsection{SISU Scenario Problem Formulation}
In this subsection, we formulate the optimization problem for the SISU case. Since optimizing with respect to continuous time will result in an infinite number of variables, we adopt the path discretization technique as in \cite{Bench1}, where the trajectory is discretized into small segments. By segmenting the trajectory into $N$ pieces, the UAV energy consumption minimization problem can be formulated as 
\begin{align}
\text{(P1):} \min_{\bT,\boldsymbol{\tau},\bq} &\sum^N_{n=1} T_nP(V_n)+\tau_nP_c \notag \\
\text{s.t.} \  &\bq[1] = \bq_0,\ \bq[N+1] = \bq_\text{F}, \label{eq: initial, f}\\
&\tau_n \leq T_n, \ \tau_n \geq 0, \label{eq: flight}\\
&\Delta_n \leq \min \{\Delta_{\max},V_{\max}T_n\},\\
&\sum^N_{n=1}\tau_n\hat{R}_n \geq Q, \label{eq:rate}
\end{align}
where $T_n, \tau_n, V_n, \bq \left[n \right]$ and $\hat{R}_n$ are the UAV flight time, UAV transmission time, UAV speed, UAV horizontal location, and achievable rate at the $n$-th segment, respectively. We assume the UAV has fixed initial and final locations $\bq_0$ and $\bq_\text{F}$, respectively, as in \eqref{eq: initial, f}. The transmission time is trivially bounded by the flight time as in \eqref{eq: flight}, and we assume that the UE has a fixed data constraint $Q$ as in \eqref{eq:rate}. One characteristic of the path discretization technique is it assumes that the channels are constant within a segment. To uphold this assumption, for the segment length $\Delta_n = \norm{\bq [n+1] - \bq[n]}$, we bound the length with $\min \{\Delta_{\max},V_{\max}T_n\}$, where $\Delta_{\max}$ is the maximum length of a segment to conserve this assumption, and $V_{\max}$ is the maximum speed of the UAV. 

Similar to (P4) in \cite{Bench1}, we transform (P1) into an equivalent problem as
\begin{align}
\text{(P2):} \min_{\bA, \by, \boldsymbol{\tau}, \bT, \bq, \bu, \bv} &\sum^N_{n=1} \tau_nP_c + P_0\left(T_n + \frac{\delta_b \Delta_n^2}{T_n} \right) + \delta_p\frac{\Delta_n^3}{T_n^2} \notag \\
&+ P_iy_n   \notag \\
\text{s.t.} \  &\bq[1] = \bq_0,\ \bq[N+1] = \bq_\text{F},\\
&\Delta_n \leq \min \{\Delta_{\max},V_{\max}T_n\},\\
&\tau_n \leq T_n, \ \tau_n \geq 0, \\
&d_{\text{I},n} \leq u_n, \ d_{\text{U},n} \leq v_n, \\
&Q \leq \sum^N_{n=1} A_n^2,\label{eq: const0}\\
&\frac{A_n^2}{\tau_n} \leq R_n, \label{eq: const1}\\
&\frac{T_n^4}{y_n^2} \leq y_n^2 + \frac{\Delta_n^2}{v_0^2}, \label{eq:const2}
\end{align}
where $\delta_b = 3U_{\text{tip}}^{-2}$, and $\delta_p = \frac{1}{2} d_0 \rho s A_\text{p}$. The transformation with regard to slack variables $A_n$ and $y_n$, which are the $n$-th elements of $\bA$ and $\by$, are equivalent as in \cite{Bench1}, thus, we omit the transformation details. The slack variables $u_n$ and $v_n$ are the $n$-th elements of the vectors $\bu$ and $\bv$, where they upper bound the UAV-IRS distance and the UAV-UE distance, respectively, which are given as 
\begin{align}
	d_{\text{I},n} &= \lVert[\bq^\text{T}[n],H_\text{A}] - [\bq_{\text{I}}^\text{T}, H_\text{I}] \rVert, \\
	d_{\text{U},n} &= \lVert[\bq^\text{T}[n],H_\text{A}] - [\bq_{\text{U}}^\text{T}, H_\text{U}] \rVert.
\end{align}
 Also, $R_n$ is the achievable rate with the slack variables $u_n$ and $v_n$ instead of $d_{\text{I},n}$ and $d_{\text{U},n}$, respectively, which can be expressed as
 \begin{align}  \label{eq: single IRS rate in opt form}
 R_n &= \sum_{\bs \in \left\{ 0,1\right\}^2} \prod_i p_i^{s_i} \left(1-p_i\right)^{1-s_i} \log_2 \left( 1+ \gamma_{\bs,n}\right), 
  \\
 \gamma_{\bs,n} &= s_\text{U} \gamma_\text{U}^2 v_n^{-\alpha_\text{U}} + 2s_\text{U} s_\text{I} \gamma_\text{U}' \gamma_\text{I}' v_n^{-\alpha_\text{U}/2}u_n^{-\alpha_\text{I}/2} \notag \\ &+ s_\text{I} \gamma_\text{I}^2 u_n^{-\alpha_\text{I}}.
 \end{align}
We observe that by considering the IRS, the achievable rate has an involved form compared to the achievable rate with no IRSs, which makes the analysis difficult. To the best of our knowledge, this is the first paper to acknowledge the achievable rate for a general number of IRSs.

While the objective function is a convex function, the constraints \eqref{eq: const0}-\eqref{eq:const2} do not follow the standard convex optimization problem form. To resolve this issue, we adopt the well known SCA technique. By deriving a global bound of a constraint that meets specific conditions, the SCA technique iteratively solves the optimization problem and guarantees to converge to a KKT condition point. One way to obtain the global lower bound satisfying the specific conditions is to apply the first-order Taylor expansion to a convex function on a local point. The right hand side of the constraints \eqref{eq: const0} and \eqref{eq:const2} have been shown to be convex in \cite{Bench1}, while $R_n$ from \eqref{eq: const1}, which is defined in \eqref{eq: single IRS rate in opt form}, needs to be analyzed. Therefore, we propose the following lemma.
\begin{lemma} \label{lemma 1}
 	For given non-negative constants $\epsilon_z,  \zeta_z,  \alpha_z$, $z \in \{1,2,..., Z\}$, the function given as 
 	\begin{align}
&f(u_1, ..., u_Z)\notag \\
&=\log_2 \left( 1+ \sum_{z=1}^{Z} \epsilon_z u_z^{-\alpha_z} + \sum_{z = 1}^{Z} \sum_{i \ne z }^Z \zeta_z \zeta_i u_z^{-\alpha_z/2}u_i^{-\alpha_i/2}  \right),
 	\end{align}
  is convex with respect to $u_z >0, z \in \{1,2,..., Z\}$ for any non-negative integer $Z$. 
\end{lemma}
\begin{proof}
	The proof is given in Appendix A.
\end{proof}
Since $p_i$ and $(1-p_i)$ are non-negative constants, $R_n$ is a weighted linear sum of the function $f(u_1, ..., u_Z)$ with different values of $Z$, thus, through Lemma \ref{lemma 1}, we confirm that $R_n$ is indeed a convex function.  Due to the convexity of $R_n$, we can successfully use the first-order Taylor expansion of $R_n$ as the lower bound. 

Adopting the lower bound of the constraints \eqref{eq: const0}-\eqref{eq:const2}, the overall optimization problem in the $\ell$-th iteration can be expressed as
\begin{align}
\text{(P3):} \min_{\bA, \by, \boldsymbol{\tau}, \bT, \bq, \bu, \bv} &\sum^N_{n=1} \tau_nP_c + P_0\left(T_n + \frac{\delta_b \Delta_n^2}{T_n} \right) + \delta_p\frac{\Delta_n^3}{T_n^2} \notag \\
&+ P_iy_n   \notag \\
\text{s.t.} \  &\bq[1] = \bq_0,\ \bq[N+1] = \bq_\text{F},\\
&\Delta_n \leq \min \{\Delta_{\max},V_{\max}T_n\},\\
&\tau_n \leq T_n, \ \tau_n \geq 0, \\
&d_{\text{I},n} \leq u_n, \ d_{\text{U},n} \leq v_n, \\
&Q \leq \sum_{n=1}^N \tilde{A}^{(\ell)}_n, \ \frac{A_n^2}{\tau_n} \leq \tilde{R}^{(\ell)}_n, \\
&\frac{T_n^4}{y_n^2} \leq \tilde{Y}^{(\ell)}_n,
\end{align}
where the lower bound functions are defined as
\begin{align}
	\tilde{A}_n^{(\ell)} &= A_n^{(\ell)2} + 2A_n^{(\ell)} \left( A_n - A_n^{(\ell)}\right), \\
	\tilde{Y}_n^{(\ell)} &= y_n^{(\ell)2} + 2y_n^{(\ell)} \left( y_n - y_n^{(\ell)}\right)- \frac{\Delta_n^{(\ell)2}}{v_0^2} \notag \\
	&+ \frac{2}{v_0^2}( \bq^{(\ell)}[n+1] -\bq^{(\ell)}[n] )^{\text{T}}\left( \bq[n+1] -\bq[n] \right), \\
	\tilde{R}_n^{(\ell)} &= R_n^{(\ell)} + \nabla R_n^{(\ell) \text{T}} \left(
	\begin{bmatrix}
	u_n - u_n^{(\ell)}\\
	v_n - v_n^{(\ell)}
	\end{bmatrix}
	\right).
\end{align}
The values $A_n^{(\ell)}, y_n^{(\ell)}, \Delta_n^{(\ell)}, \bq^{(\ell)}[n], u_n^{(\ell)}, v_n^{(\ell)},$ and $R_n^{(\ell)}$ are the local points to obtain the lower bound, which are the results of the $\left(\ell -1\right)$-th iteration, and the gradient $ \nabla R_n$ can be calculated through the result of Appendix B. Note that, for the first iteration, we must define an initial case to compute the problem (P3). We will specify the initial case in Section~\ref{sec: simul}.

\begin{algorithm}[t]
	\begin{algorithmic} [1]
		\caption{Pseudo-code for the SISU scenario UAV energy consumption minimization algorithm}
		\State \textbf{Initialization:} Set $A_n^{(0)}, y_n^{(0)}, \Delta_n^{(0)}, \bq^{(0)}[n], u_{n}^{(0)}, v_{n}^{(0)},$ and $R_{n}^{(0)}$. Let $\ell = 0$.
		\Repeat
		\State Set $\ell = \ell +1$.
		\State Solve the optimization problem (P3).
		\State Update the solution of (P3) as $A_n^{(\ell)}, y_n^{(\ell)}, \Delta_n^{(\ell)},$ $\bq^{(\ell)}[n], u_{n}^{(\ell)}, v_{n}^{(\ell)},$ and $R_{n}^{(\ell)}$.
		\Until \ The UAV energy consumption $E_\text{tot}$ decreases by a fraction below a predefined threshold.
	\end{algorithmic}
\end{algorithm} 

Since the lower bounded constraints are all affine functions, the problem (P3)  is now in the standard convex optimization problem form, thus, can be solved effectively via problem solving tools such as CVX \cite{CVX1}. By iteratively solving the optimization problem and updating the local points, the solutions of (P3) will be in the feasible region of the original problem (P2), which converges to a KKT condition point. The proposed optimization technique is summarized in Algorithm 1. Note that, since the optimization problem is from the upper bound in \eqref{eq: jensen}, the resulting solution may not satisfy the data constraint in reality. To settle this issue, we can give a margin to the UE constraint $Q$, e.g., $Q+\chi$ for some positive constant $\chi$. Thus, by restraining the constraint, it will have a larger probability to satisfy the constraint in practice.

\subsection{SIMU Scenario Problem Formulation}
In this subsection, we formulate the optimization problem to minimize the UAV energy consumption in a SIMU scenario. However, this can be performed very easily with the similar expression as in \cite{Bench1}. The resulting optimization problem is given as
\begin{align}
\text{(P4):} \min_{\bT,\boldsymbol{\tau},\bq} &\sum_{n=1}^N \left\{T_nP(V_n)+\sum_{k=1}^K \tau_{nk}P_c \right\}\\
\text{s.t.} \  &\bq[1] = \bq_0,\ \bq[N+1] = \bq_\text{F},\\
&\sum_{k=1}^K \tau_{nk} \leq T_n, \ \tau_{nk} \geq 0, \label{eq: tdma} \\
&\sum^N_{n=1}\tau_{nk}\hat{R}_{nk} \geq Q_k, \label{eq: UE rates}\\
&\Delta_n \leq \min \{\Delta_{\max},V_{\max}T_n\},
\end{align}
where $\tau_{nk}$ and $\hat{R}_{nk}$ are the transmit time with the $k$-th UE at the $n$-th segment and the achievable rate for the $k$-th UE at the $n$-th segment, respectively, and $Q_k$ is the data constraint for the $k$-th UE.   Since we assume that the channels are approximately constant within a segment, the TDMA assumption is legitimate with the constraint \eqref{eq: tdma}. With the multiple UE expansion, we observe that the transmission for all UEs must be considered in each flight time as in \eqref{eq: tdma}, and the data constraint for each UE must be sufficed independently as in \eqref{eq: UE rates}. We neglect the transformation process of (P4) due to redundancy. However, (P4) will be a subset within the general MIMU case, thus, can be effectively solved with the latter solutions in Section \ref{sec: multi IRS}.

\section{Multiple IRS Optimization}\label{sec: multi IRS}
In this section, we consider the most general scenario of MIMU case. To minimize the UAV energy consumption, we propose two optimization problems to effectively minimize the UAV energy consumption with similar performance. In addition, we propose an algorithm to mimic the behavior of a specific optimization problem to achieve low UAV energy consumption with minimal complexity.

\subsection{General IRS Usage Case}
In the MIMU scenario, the most straightforward approach to exploit the multiple IRSs is to use all of them simultaneously. We denote this approach as the general IRS usage case. In order to realize this concept, we must define the UE achievable rate with regard to multiple IRSs. Fortunately, due to TDMA, only one UE is considered in a time instance, thus,
the phase shifts of the IRSs can be easily obtained with the same approach as \eqref{eq: phase shift}. Assuming that the IRSs concentrate on the $k$-th UE, the achievable rate is obtained as
\begin{align}
	\hat{R}_k(t) &= \sum_{\bs \in \left\{ 0,1\right\}^{W+1}} \prod_i p_i^{s_i} \left(1-p_i\right)^{1-s_i} \log_2 \left( 1+ \hat{\gamma}_\bs (t)\right), \\
	& i \in \{\text{U}_k, \text{I}_1, ... , \text{I}_W \}, \notag
\end{align}
with the state vector as $\bs = \left[ s_{\text{U}_k}, s_{\text{I}_1}, ... , s_{\text{I}_W} \right]^{\text{T}}$, and the SNR as
\begin{align}
\hat{\gamma}_\bs(t) &= \sum_{i} s_i \gamma_i^2 d_i^{-\alpha_i} (t) \notag \\
&+ \sum_{i} \sum_{j \ne i } s_i s_j \gamma_i' \gamma_j' d_i^{-\alpha_i/2} (t)d_j^{-\alpha_j/2}(t),
\end{align}
where $\gamma_i$ and $\gamma_i'$ follow the same definitions as \eqref{eq: gam1}-\eqref{eq: gam4}.

In result, the optimization problem will have the same expression as (P4), and the problem can be transformed in the equivalent form as
\begin{align}
\text{(P5):} \min_{\bA, \by, \boldsymbol{\tau}, \bT, \bq, \bu, \bv} &\sum^N_{n=1} \sum_{k=1}^{K} \tau_{nk}P_c + P_0\left(T_n + \frac{\delta_b \Delta_n^2}{T_n} \right)\notag \\
 &+ \delta_p\frac{\Delta_n^3}{T_n^2} + P_iy_n   \notag \\
\text{s.t.} \  &\bq[1] = \bq_0,\ \bq[N+1] = \bq_\text{F},\\
&\Delta_n \leq \min \{\Delta_{\max},V_{\max}T_n\},\\
&\sum_{k=1}^K \tau_{nk} \leq T_n, \ \tau_{nk} \geq 0, \\
&d_{\text{I}_w,n} \leq u_{nw}, \ d_{\text{U}_k,n} \leq v_{nk}, \\
&Q_k \leq \sum^N_{n=1} A_{nk}^2,\\
&\frac{A_{nk}^2}{\tau_{nk}} \leq R_{nk}, \\
&\frac{T_n^4}{y_n^2} \leq y_n^2 + \frac{\Delta_n^2}{v_0^2},
\end{align}
where we observe that the slack variables $v_{nk}, A_{nk}$, and $u_{nw}$ have an additional index due to the MIMU generalization. Using Lemma \ref{lemma 1}, we have shown that $R_{nk}$ is a convex function with respect to $u_{nw}$ and $v_{nk}$ for a general number of IRSs. Thus, we can effectively solve (P5) with the same method as (P2). 

\subsection{IRS Matching Case}
While exploiting all the IRSs simultaneously might be the most straightforward method, it may not be the most probable solution. Considering that the IRS is a communication assisting object, the IRSs are likely to be located far apart. Thus, due to the double propagation characteristic of the IRS, the power from the IRSs that are far apart from a UE is likely to be negligible. Also, the assumption that every IRS is controlled simultaneously for a single UE may be implausible. Thus, the resulting solutions may be less practical.  

Accordingly, we consider a more practical case that each UE is supported by only one IRS, where we denote it as the IRS matching case. The formulated optimization problem will effectively evaluate and decide which IRS to use per instance. The overall problem can be expressed as
\begin{align}
\text{(P6):}\min_{\bT,\boldsymbol{\tau},\bq} &\sum_{n=1}^N \left\{T_nP(V_n)+\sum_{k=1}^K \tau_{nk}P_c \right\}\\
\text{s.t.} \  &\bq[1] = \bq_0,\ \bq[N+1] = \bq_\text{F},\\
&\sum_{k=1}^K \tau_{nk} \leq T_n, \ \tau_{nk} \geq 0, \\
&\sum_{n=1}^N\tau_{nk}\max_w \hat{R}_{nwk} \geq Q_k, \label{eq: max in} \\
&\Delta_n \leq \min \{\Delta_{\max},V_{\max}T_n\},
\end{align}
 where $\hat{R}_{nwk}$ represents the achievable rate of the $k$-th UE at the $n$-th segment when the UE is matched with the $w$-th IRS. Hence, with the maximum function in \eqref{eq: max in}, we pick the best IRS per instance for communication. Since we exploit only a single IRS, $\hat{R}_{nwk}$ is equivalent to \eqref{eq: single IRS rate} by considering the signals from only the $w$-th IRS, where we assume the signals from other IRSs are negligible. While we can still adopt the SCA technique due to the convexity of $\max_w \hat{R}_{nwk}$, the gradient is quite difficult to compute. Instead, we introduce a matching variable $\eta_{nwk}$, where it represents the communication time between the $k$-th UE in the $n$-th segment using the $w$-th IRS. The overall problem is given as
\begin{align}
\text{(P7):} \min_{\bT,\boldsymbol{\tau},\bq, \boldsymbol{\eta}} &\sum_{n=1}^N \left\{T_nP(V_n)+\sum_k \tau_{nk}P_c \right\}\\
\text{s.t.} \  &\bq[1] = \bq_0,\ \bq[N+1] = \bq_\text{F},\\
&\sum_{k=1}^K \tau_{nk} \leq T_n, \ \tau_{nk} \geq 0, \\
&\sum_{w=1}^W \eta_{nwk} \leq \tau_{nk}, \ \eta_{nwk} \geq 0, \label{eq: IRS match}\\
&\sum_{n,w}\eta_{nwk}\hat{R}_{nwk} \geq Q_k, \label{eq: IRS select} \\
&\Delta_n \leq \min \{\Delta_{\max},V_{\max}T_n\},
\end{align}
where with the additional constraints \eqref{eq: IRS match} and \eqref{eq: IRS select}, the solution will effectively force only one $\eta_{nwk}$, which is associated with the largest $\hat{R}_{nwk}$, to $\tau_{nk}$, and the other $\eta_{nw'k}$ values to zero. By adopting the matching variable $\eta_{nwk}$, we successfully removed the maximum operation in \eqref{eq: max in}, and the gradient of $\hat{R}_{nwk}$ can be calculated through the results from the previous section since $\hat{R}_{nwk}$ has the same expression as \eqref{eq: single IRS rate}. As we can observe, (P7) is the same as (P4) with the additional linear constraint \eqref{eq: IRS match}, thus, we can solve the problem equivalently through the SCA technique. Note that, by choosing a single IRS for each instance, the IRS matching case can also compensate for the rather strong assumption that the LoS path between the IRSs and UEs always exist, since IRSs closely located to a UE will have a high probability to have an LoS path. Thus, the IRS matching case is more practical than the general IRS usage case.

\subsection{Low Complexity Algorithm}
In this subsection, we propose a low complexity algorithm that attempts to follow the behavior of the optimized solutions described in the previous subsection. Through the simulation results in Section \ref{sec: simul}, we observe several factors. First, the IRS matching case has marginal performance loss than the general IRS usage case. Second, when the data constraint is small, the UAV directly goes to the final location with speed $V_e$, where $V_e$ is the speed with the minimum energy consumption per meter, i.e., $V_e = \min_V P(V)/V$. Finally, when the data constraint is large, the UAV tends to communicate at fixed locations. Accordingly, to actively adapt to these characteristics, we introduce the following low complexity algorithm.

\begin{algorithm}[t]
	\begin{algorithmic} [1]
		\caption{Pseudo-code for the low complexity UAV energy consumption minimization algorithm}
		\State Pair each UE to its closest IRS
		\State Find $\hat{\bq}_k$ and $\bar{\bq}_k$.
		\State Through the toy trajectory, find $\bff$.
		\If {f_k \geq 1, \  \forall k,}
		\State Set the toy trajectory as the solution.
		\Else
		\State Compute \eqref{eq:find star}.
		\State \ \ \ \ Solve the TSP of $\left\{\bq_k^{\star}\right\}$ to find the shortest trajectory.
		\State \ \ \ \ Set the shortest trajectory as the solution.
		\EndIf
		\State \textbf{end if}
	\end{algorithmic}
\end{algorithm} 

For this algorithm, we pair each UE with the closest IRS to mimic the IRS matching case, given the fact that the IRS matching case can suffice considerable performance. Before we specify the trajectory, we define several auxiliary variables. Through a one dimensional search between each UE and its paired IRS, we find the largest achievable rate location for each UE, denoted as $\hat{\bq}_k$. Afterwards, we draw a straight line from $\bq_0$ to $\bq_\text{F}$, denoted as a toy trajectory. Next, we find the closest point from $\hat{\bq}_k$ to the toy trajectory as $\bar{\bq}_k$ by landing a vertical line on the toy trajectory from $\hat{\bq}_k$. Using these variables, we will acquire the trajectory for the UAV.

First, we simulate a toy scenario which follows the toy trajectory, moves with speed $V_e$, and performs equal time transmission for each UE at every segment, i.e., $\tau_{nk}= T_n/K$. From this toy scenario, we can compute the transmitted data as
\begin{align}
\br = \left[ f_1Q_1, ..., f_KQ_K\right]^{\text{T}},
\end{align} 
where $f_k$ is the fraction of data the UAV has transmitted with respect to the $k$-th UE with data constraint $Q_k$. We can see that the fraction $f_k$ indirectly implies the magnitude of the data constraint and distance between the toy trajectory and the UE, e.g., $f_k$ will be small if $Q_k$ is too large or the achievable rate is too small due to the large UE distance. Note that, if $f_k \geq 1, \  \forall k \in \left\{1,2,...,K\right\}$, we can simply select the toy scenario as the algorithm solution. Otherwise, we fix the transmit locations as
\begin{align} \label{eq:find star}
	\bq^{\star}_k = \hat{\bq}_k + g_k(\bff) \left(\bar{\bq}_k - \hat{\bq}_k\right),
\end{align}
where $\bq^{\star}_k$ denotes the transmit location for the $k$-th UE, $\bff = \left[ f_1, ..., f_K\right]^\text{T}$, and $g_k(\bff)$ is a well defined function to choose the internal division point between $\bar{\bq}_k$ and $\hat{\bq}_k$ using $\bff$. With the transmit locations $\{\bq^{\star}_k\}$, we solve the traveling salesman problem (TSP) to find the shortest trajectory \cite{Bench1}. Finally, the UAV moves with speed $V_e$ and the transmission takes place while the UAV is moving, and while the UAV is at the hovering locations. In specific, we first calculate the amount of transmitted data while the UAV is moving, and the remaining data is transmitted while hovering at the fixed locations $\{\bq_k^\star\}$. In conclusion, the UAV will fly straight to $\bq_\text{F}$ for small data constraints, and will fly to all the designated transmit locations $\{ \bq^{\star}_k\}$ in straight lines for large data constraints. In this paper, we define $g_k(\bff) = \min(f_k,1)$. While determining the optimal $g_k (\bff)$ is itself another area of study, we consider it outside the scope of this paper. Nevertheless, the simple $g_k(\bff)$ is shown to have considerable performance in Section \ref{sec: simul}. The overall algorithm is summarized in Algorithm 2.

\section{Simulation Results} \label{sec: simul}
In this section, we verify and compare the performance of the proposed algorithms. For simulations, the altitude and heights are fixed as $H_\text{A} = 100$ m, $H_{\text{I}_w} = 20$ m, and $H_{\text{U}_k} = 0$, and the number of IRS elements is fixed as $M_w = 500$. For the probabilistic LoS model, the parameters are fixed as $a_\text{U} = 30$, $b_\text{U} = 0.15$, $a_\text{I} = 15$, $b_\text{I} = 0.18$, $\theta_\text{U} = 60^\circ$, and $\theta_\text{I} = 54.2^\circ$. For the channel model, the coefficients are fixed as $\beta_0 = 0.01$, $\alpha_\text{I} = 2.2$, $\alpha_\text{U} = 2.5$, $\alpha = 3$, $\kappa_\text{I} = 30$, $\kappa_\text{U}=10$, $\kappa = 5$, $\sigma^2 = -174$ dBm, and $B = 1$ MHz. For the communication scenario, the parameters are fixed as $\bq_0 = [0,0]^\text{T}$, $\bq_\text{F} = [100,100] ^\text{T}$, $\Delta_{\max} = 1$ m, $V_e = 18.3$ m/s, $V_{\max} = 30$ m/s, and we assume that $Q_k$ is the same for all UEs for simplicity. We construct a simple solution, i.e., hover and transmit at the UE locations, and use it for the initial case for the algorithms adopting the SCA technique. In the figures, the UEs are denoted as black circles and the IRSs are denoted as black diamonds.

Additionally, we simulate the no IRS case as a benchmark by setting $p_{\text{I}_w}= 0$, which will have the same result as \cite{Bench1}. While other UAV trajectory optimization techniques exist for a multiple IRS scenario, there were no results considering UAV energy consumption. Thus, we only adopt the no IRS case as a benchmark.

\begin{figure}[t] 
	\centering
	\includegraphics[width=1 \columnwidth]{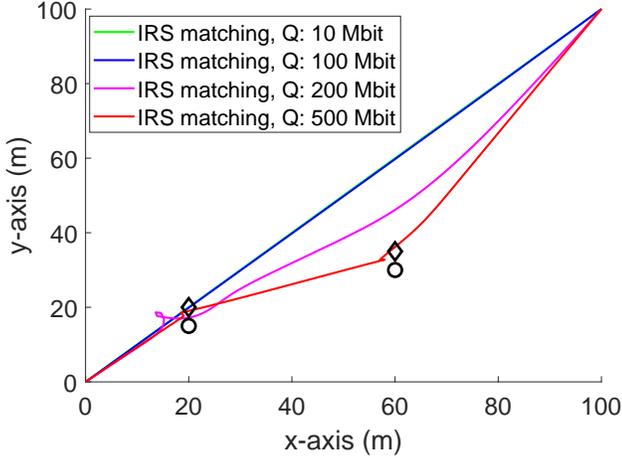}
	%   % where an .eps filename suffix will be assumed under latex,
	%   % and a .pdf suffix will be assumed for pdflatex
	\caption{UAV trajectories for the IRS matching case for different data constraints.} 
	\label{fig:About matching}
\end{figure}

In Fig. \ref{fig:About matching}, we observe the UAV trajectory for the IRS matching case with different data constraints. The figure shows that the UAV trajectory follows a straight line for a small data constraint, and steadily shifts towards the UE locations as the data constraint increases. These results follow our intuition. When the data constraint is small, the UAV can easily satisfy the data constraint from a far distance, thus, it uses the minimum flight power by simply flying on a straight line. As the data constraint increases, the UAV must find a balance between the energy consumption due to flying and the increasing achievable rate from shorter distance between the UAV and UEs. Thus, as the rate constraint increases, there is an increasing need for high achievable rate, resulting in a trajectory shift towards the UEs. Note that, we observe some unnatural behavior in some of the trajectories such as loops. This phenomenon seems to occur due to the UAV power consumption model, i.e., the UAV power consumption is not minimized at $V=0$, therefore; the UAV may need to keep moving near the UEs to reduce its power consumption unless the data constraint is too high. Note also that the solution of the optimization problem only converges to a local optimum since the problem is not in the standard convex optimization form.

\begin{figure}[t] 
	\centering
	\includegraphics[width=1 \columnwidth]{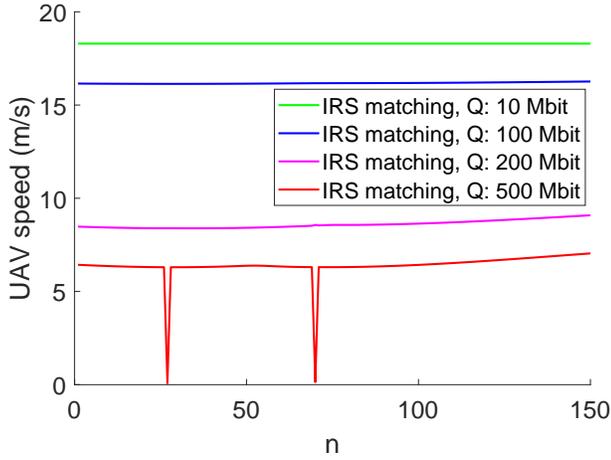}
	%   % where an .eps filename suffix will be assumed under latex,
	%   % and a .pdf suffix will be assumed for pdflatex
	\caption{UAV flight speed with respect to the path segment $n$ for the IRS matching case with different data constraints.} 
	\label{fig:AboutV}
\end{figure}

\begin{figure}[t] 
	\centering
	\includegraphics[width=1 \columnwidth]{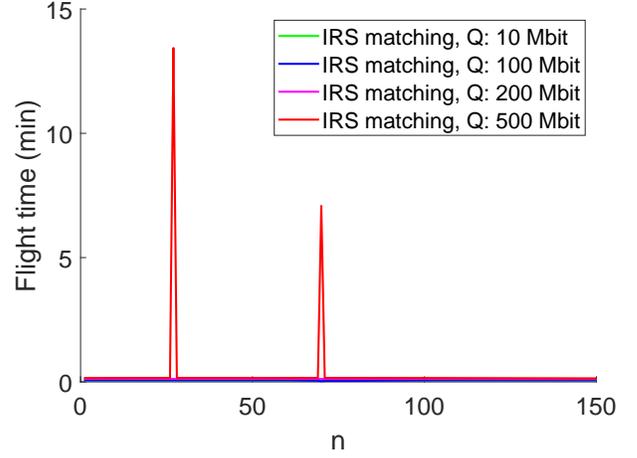}
	%   % where an .eps filename suffix will be assumed under latex,
	%   % and a .pdf suffix will be assumed for pdflatex
	\caption{UAV flight time with respect to the path segment $n$ for the IRS matching case with different data constraints.} 
	\label{fig:AboutT}
\end{figure}

In Figs. \ref{fig:AboutV} and \ref{fig:AboutT}, we plot the UAV speed and UAV flight time with respect to the path discretization segment $n$ for the IRS matching case with the scenario equivalent to Fig. \ref{fig:About matching}. Note that, the flight time for the small data constraint cases have minimal flight time, thus, are all shown at the bottom of Fig. \ref{fig:AboutT}. We first observe that the UAV speed for the lowest data constraint case is a constant with value $V_e$. This is expected, since, to reduce UAV energy consumption, the UAV must choose not only the shortest trajectory, but also the optimal speed for energy efficiency. Also, by jointly observing the speed and flight time, we find out that as the data constraint increases, the UAV has the tendency to slow down and communicate on a fixed location, which is the basis of our proposed low complexity algorithm.

\begin{figure}[t] 
	\centering
	\includegraphics[width=1 \columnwidth]{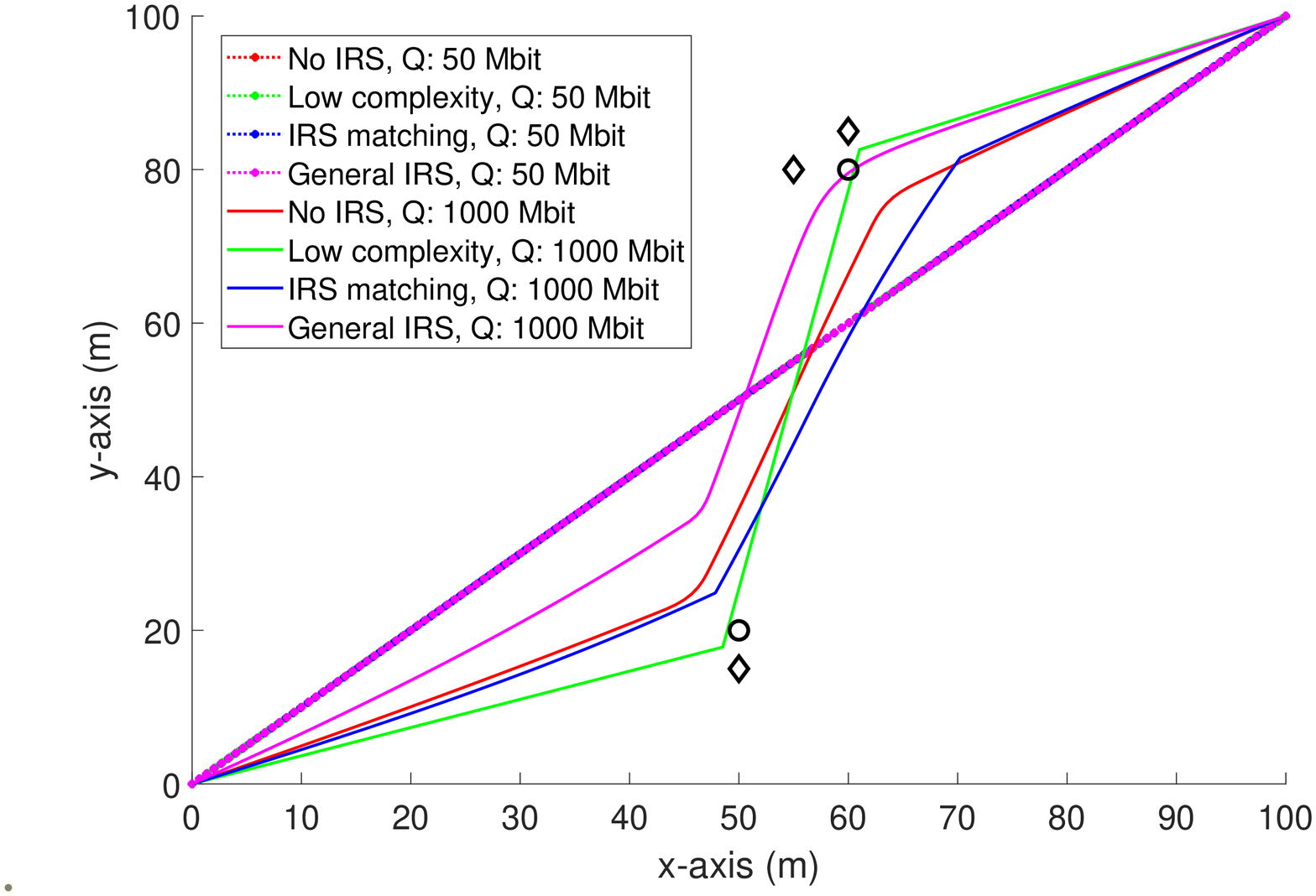}
	%   % where an .eps filename suffix will be assumed under latex,
	%   % and a .pdf suffix will be assumed for pdflatex
	\caption{UAV trajectories in extreme data constraints.} 
	\label{fig:AboutAlg}
\end{figure}

In Fig. \ref{fig:AboutAlg}, we plot the trajectory of the UAV for the various proposed techniques in two extreme data constraints. For the low data constraint, we observe that all the techniques converge to the same trajectory, confirming that the algorithms work in a consistent manner for low data constraints. Also, it confirms that the proposed low complexity algorithm adapts well to the low data constraint. For the high data constraint, we observe that the trajectories behave in a more involved manner, confirming that proper adjustment is needed for different communication requirements.

\begin{figure}[t] 
	\centering
	\includegraphics[width=1 \columnwidth]{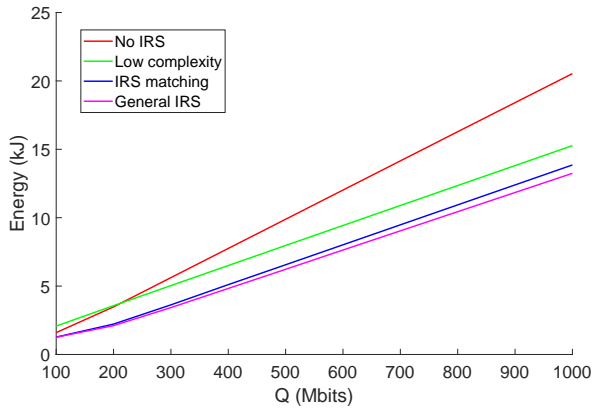}
	%   % where an .eps filename suffix will be assumed under latex,
	%   % and a .pdf suffix will be assumed for pdflatex
	\caption{UAV energy consumption with respect to the data constraint.} 
	\label{fig:AboutEE}
\end{figure}

In Fig. \ref{fig:AboutEE}, we plot the UAV energy consumption for different techniques with respect to the data constraint for the scenario equivalent to Fig. \ref{fig:AboutAlg}. Note that, since one UE has two IRSs in close proximity, the results will be beneficial to the general IRS usage case. As expected, the most prominent result is that the proposed optimization problem solutions have outstanding performance with respect to the no IRS case, confirming that exploiting the IRS for UAV energy consumption minimization is a legitimate approach. Also, the general IRS usage case works as a lower bound for any data constraint. This result is expected since the general IRS usage case exploits all the IRSs in the system, thus, will have the best performance. We also observe that the performance of the IRS matching case declines faster than the general IRS usage case. This is due to the composition of the UAV energy consumption model. While for small data constraints, the UAV energy consumption is dominated by the flying power to the final location, which is inevitable. However, as the data constraint increases, the UAV hovers in a fixed location for communication, thus, the energy consumption from hovering for communication becomes dominant. Since the general IRS usage case will have the largest achievable rate, the energy consumption of hovering for communication will be generally smaller, resulting in better performance than the IRS matching case. Finally, we observe the performance of the proposed low complexity algorithm. In the low data constraint region, the low complexity algorithm has some performance loss due to the constant speed of the UAV. However, the low complexity algorithm has marginal performance loss with respect to the IRS matching case in the high data constraint region. Indeed, we have confirmed that with extremely high data constraints, the performance of the low complexity case converges to the IRS matching case. This confirms that the proposed low complexity algorithm works sufficiently well in the high data constraint region with a simply defined $g_k(\bff)$, while it can be further improved by optimizing the speed of UAV for the low data constraint region, which is an interesting future research topic.

\section{Conclusion} \label{sec:concl}
In this paper, we proposed several techniques to minimize the UAV energy consumption in a general MIMU scenario. We successfully developed the probabilistic LoS channel model to ensure the advantages of highly located IRSs. Through some approximations, we transformed and derived the new system into a tractable form, which could be solved by the SCA technique. Employing various approaches, we derived an algorithm that exploits all the IRSs simultaneously and an algorithm that actively selects the best IRS for transmission. Using the results, we proposed a low complexity algorithm that follows its performance. By comparison with the no IRS case, we confirmed that exploiting the IRSs in UAV communication systems is indeed favorable to minimize the UAV energy consumption.

\section*{Appendix A \\ Proof of Lemma 1}
Note that the functions
\begin{align}
	\log_2 \left( \epsilon_z u_z^{-\alpha_z}\right) &=-\alpha_z \log_2\left( u_z\right) + \log_2 \left( \epsilon_z \right), \\
\log_2 \left( \zeta_z \zeta_i u_z^{-\alpha_z/2} u_i^{-\alpha_i/2}\right) &=-\frac{\alpha_z}{2} \log_2\left( u_z\right)  \notag \\
\quad  &-\frac{\alpha_i}{2}\log_2\left( u_i\right) 
+\log_2 \left( \zeta_z \right) \notag \\
&+\log_2 \left( \zeta_i \right),
\end{align}
are convex with respect to $u_z$ and $u_i$, since the logarithmic function is concave and $\alpha_z \geq 0$. Thus, the functions $\epsilon_z u_z^{-\alpha_z}$ and $ \zeta_z \zeta_i u_z^{-\alpha_z/2} u_i^{-\alpha_i/2}$ are log-convex functions. With the same approach, we can also find out that a constant value is also a log-convex function.

By exploiting the fact that the sum of log-convex functions is log-convex, we see that the function $f(u_1, ..., u_Z)$ is the sum of log-convex functions, which finishes the proof.

\section*{Appendix B}
For the function 
\begin{align}
&f(u_1, ..., u_Z)   \notag \\
&=\log_2 \left( 1+ \sum_{z=1}^{Z} \epsilon_z u_z^{-\alpha_z} + \sum_{z = 1}^{Z} \sum_{i \ne z}^Z \zeta_z \zeta_i u_z^{-\alpha_z/2}u_i^{-\alpha_i/2}  \right)  \notag \\
&= \log_2 \left( g(u_1, ..., u_Z)\right),
\end{align}
the derivative with respect to $u_z$ is given as
\begin{align}
&\frac{df}{du_z} =\frac{-\alpha_z \epsilon_z u_z^{-\alpha_z -1} - \sum_{i \ne z}^Z -\alpha_z \zeta_z \zeta_i u_z^{-\alpha_z/2 -1} u_i^{-\alpha_i/2}}{g(u_1,...,u_Z) \ln 2}.
\end{align}
In result, the derivative of $f(u_1,...,u_Z)$ is expressed as
\begin{align}
	\nabla f = \left[ \frac{df}{du_1}, ... , \frac{df}{du_Z} \right] ^\text{T},
\end{align}
which finishes the derivation.

\section*{Acknowledgement}
The authors thank S. Kang and B. Ko at KAIST for helping with the derivation of Appendix A and designing Fig. \ref{fig:System model}.
%\appendix
%%%%%%%%%%%%%%%%
%% Background %%
%%

\bibliographystyle{IEEEtran}
% argument is your BibTeX string definitions and bibliography database(s)
\bibliography{HyebibUAV}

% Generated by IEEEtran.bst, version: 1.14 (2015/08/26)
\begin{thebibliography}{10}
\providecommand{\url}[1]{#1}
\csname url@samestyle\endcsname
\providecommand{\newblock}{\relax}
\providecommand{\bibinfo}[2]{#2}
\providecommand{\BIBentrySTDinterwordspacing}{\spaceskip=0pt\relax}
\providecommand{\BIBentryALTinterwordstretchfactor}{4}
\providecommand{\BIBentryALTinterwordspacing}{\spaceskip=\fontdimen2\font plus
\BIBentryALTinterwordstretchfactor\fontdimen3\font minus
  \fontdimen4\font\relax}
\providecommand{\BIBforeignlanguage}[2]{{%
\expandafter\ifx\csname l@#1\endcsname\relax
\typeout{** WARNING: IEEEtran.bst: No hyphenation pattern has been}%
\typeout{** loaded for the language `#1'. Using the pattern for}%
\typeout{** the default language instead.}%
\else
\language=\csname l@#1\endcsname
\fi
#2}}
\providecommand{\BIBdecl}{\relax}
\BIBdecl

\bibitem{mmWave1}
F.~Boccardi, R.~W. Heath, A.~Lozano, T.~L. Marzetta, and P.~Popovski, ``{Five
  Disruptive Technology Directions for 5G},'' \emph{IEEE Communications
  Magazine}, vol.~52, no.~2, pp. 74--80, Feb. 2014.

\bibitem{mmWave2}
T.~S. Rappaport, S.~Sun, R.~Mayzus, H.~Zhao, Y.~Azar, K.~Wang, G.~N. Wong,
  J.~K. Schulz, M.~Samimi, and F.~Gutierrez, ``{Millimeter Wave Mobile
  Communications for 5G Cellular: It Will Work!}'' \emph{IEEE Access}, vol.~1,
  pp. 335--349, May 2013.

\bibitem{mmWave3}
W.~Roh, J.~Seol, J.~Park, B.~Lee, J.~Lee, Y.~Kim, J.~Cho, K.~Cheun, and
  F.~Aryanfar, ``{Millimeter-wave Beamforming as an Enabling Technology for 5G
  Cellular Communications: Theoretical Feasibility and Prototype Results},''
  \emph{IEEE Communications Magazine}, vol.~52, no.~2, pp. 106--113, Feb. 2014.

\bibitem{mmPotential1}
Y.~Niu, Y.~Li, D.~Jin, L.~Su, and A.~V. Vasilakos, ``{A Survey of Millimeter
  Wave Communications (mmWave) for 5G: Opportunities and Challenges},''
  \emph{Wireless Networks}, vol.~21, pp. 2657--2676, Apr. 2015.

\bibitem{mmPotential2}
S.~Rangan, T.~S. Rappaport, and E.~Erkip, ``{Millimeter-Wave Cellular Wireless
  Networks: Potentials and Challenges},'' \emph{Proceedings of the IEEE}, vol.
  102, no.~3, pp. 366--385, Feb. 2014.

\bibitem{IRS1}
W.~{Saad}, M.~{Bennis}, and M.~{Chen}, ``{A Vision of 6G Wireless Systems:
  Applications, Trends, Technologies, and Open Research Problems},'' \emph{IEEE
  Network}, vol.~34, no.~3, pp. 134--142, Oct. 2020.

\bibitem{IRS2}
E.~Björnson, L.~Sanguinetti, H.~Wymeersch, J.~Hoydis, and T.~L. Marzetta,
  ``{Massive MIMO is a Reality—What is Next?: Five Promising Research
  Directions for Antenna Arrays},'' \emph{Digital Signal Processing}, vol.~94,
  pp. 3 -- 20, Nov. 2019.

\bibitem{IRS3}
H.~Cho and J.~Choi, ``{Alternating Beamforming With Intelligent Reflecting
  Surface Element Allocation},'' \emph{IEEE Wireless Communications Letters},
  vol.~10, no.~6, pp. 1232--1236, Mar. 2021.

\bibitem{IRSwork1}
M.~Renzo, M.~Debbah, and D.~PhanHuy, ``{Smart Radio Environments Empowered by
  Reconfigurable AI Meta-Surfaces: An idea whose time has come.}''
  \emph{EURASIP Journal on Wireless Com Network}, no. 129, May 2019.

\bibitem{IRSwork2}
C.~{Liaskos}, S.~{Nie}, A.~{Tsioliaridou}, A.~{Pitsillides}, S.~{Ioannidis},
  and I.~{Akyildiz}, ``{A New Wireless Communication Paradigm Through
  Software-Controlled Metasurfaces},'' \emph{IEEE Communications Magazine},
  vol.~56, no.~9, pp. 162--169, Sep. 2018.

\bibitem{IRSwork3}
E.~{Basar}, M.~Renzo, J.~Rosny, M.~{Debbah}, M.~{Alouini}, and R.~{Zhang},
  ``{Wireless Communications Through Reconfigurable Intelligent Surfaces},''
  \emph{IEEE Access}, vol.~7, pp. 116\,753--116\,773, Aug. 2019.

\bibitem{IRSwork4}
H.~Guo, Y.~Liang, J.~Chen, and E.~G. Larsson, ``{Weighted Sum-Rate Maximization
  for Reconfigurable Intelligent Surface Aided Wireless Networks},'' \emph{IEEE
  Transactions on Wireless Communications}, vol.~19, no.~5, pp. 3064--3076,
  Feb. 2020.

\bibitem{UAV1}
Z.~Xiao, P.~Xia, and X.~Xia, ``{Enabling UAV Cellular With Millimeter-Wave
  Communication: Potentials and Approaches},'' \emph{IEEE Communications
  Magazine}, vol.~54, no.~5, pp. 66--73, May 2016.

\bibitem{UAV2}
Y.~Zeng, R.~Zhang, and T.~J. Lim, ``{Wireless Communications With Unmanned
  Aerial Vehicles: Opportunities and Challenges},'' \emph{IEEE Communications
  Magazine}, vol.~54, no.~5, pp. 36--42, May 2016.

\bibitem{UAV3}
E.~P. de~Freitas, T.~Heimfarth, I.~F. Netto, C.~E. Lino, C.~E. Pereira, A.~M.
  Ferreira, F.~R. Wagner, and T.~Larsson, ``{UAV Relay Network to Support WSN
  Connectivity},'' \emph{International Congress on Ultra Modern
  Telecommunications and Control Systems}, pp. 309--314, Dec. 2010.

\bibitem{UAV4}
B.~Ji, Y.~Li, B.~Zhou, C.~Li, K.~Song, and H.~Wen, ``{Performance Analysis of
  UAV Relay Assisted IoT Communication Network Enhanced With Energy
  Harvesting},'' \emph{IEEE Access}, vol.~7, pp. 38\,738--38\,747, Mar. 2019.

\bibitem{UAVBS1}
M.~Alzenad, A.~El-Keyi, F.~Lagum, and H.~Yanikomeroglu, ``{3-D Placement of an
  Unmanned Aerial Vehicle Base Station (UAV-BS) for Energy-Efficient Maximal
  Coverage},'' \emph{IEEE Wireless Communications Letters}, vol.~6, no.~4, pp.
  434--437, May 2017.

\bibitem{UAVBS2}
V.~Saxena, J.~Jaldén, and H.~Klessig, ``{Optimal UAV Base Station Trajectories
  Using Flow-Level Models for Reinforcement Learning},'' \emph{IEEE
  Transactions on Cognitive Communications and Networking}, vol.~5, no.~4, pp.
  1101--1112, Oct. 2019.

\bibitem{UAVrelay1}
S.~Zhang, H.~Zhang, Q.~He, K.~Bian, and L.~Song, ``{Joint Trajectory and Power
  Optimization for UAV Relay Networks},'' \emph{IEEE Communications Letters},
  vol.~22, no.~1, pp. 161--164, Oct. 2018.

\bibitem{UAVrelay2}
X.~Chen, X.~Hu, Q.~Zhu, W.~Zhong, and B.~Chen, ``{Channel Modeling and
  Performance Analysis for UAV Relay Systems},'' \emph{China Communications},
  vol.~15, no.~12, pp. 89--97, Dec. 2018.

\bibitem{UAVenergy1}
Y.~Zeng, Q.~Wu, and R.~Zhang, ``{Accessing From the Sky: A Tutorial on UAV
  Communications for 5G and Beyond},'' \emph{Proceedings of the IEEE}, vol.
  107, no.~12, pp. 2327--2375, Dec. 2019.

\bibitem{UAVenergy2}
Y.~Zeng and R.~Zhang, ``{Energy-Efficient UAV Communication With Trajectory
  Optimization},'' \emph{IEEE Transactions on Wireless Communications},
  vol.~16, no.~6, pp. 3747--3760, Mar. 2017.

\bibitem{Bench1}
Y.~Zeng, J.~Xu, and R.~Zhang, ``{Energy Minimization for Wireless Communication
  With Rotary-Wing UAV},'' \emph{IEEE Transactions on Wireless Communications},
  vol.~18, no.~4, pp. 2329--2345, Mar. 2019.

\bibitem{UAVIRS1}
S.~Li, B.~Duo, X.~Yuan, Y.~Liang, and M.~D. Renzo, ``{Reconfigurable
  Intelligent Surface Assisted UAV Communication: Joint Trajectory Design and
  Passive Beamforming},'' \emph{IEEE Wireless Communications Letters}, vol.~9,
  no.~5, pp. 716--720, Jan. 2020.

\bibitem{UAVIRS2}
D.~Ma, M.~Ding, and M.~Hassan, ``{Enhancing Cellular Communications for UAVs
  via Intelligent Reflective Surface},'' \emph{2020 IEEE Wireless
  Communications and Networking Conference (WCNC)}, pp. 1--6, May 2020.

\bibitem{UAVIRS3}
Z.~Wei, Y.~Cai, Z.~Sun, D.~W.~K. Ng, J.~Yuan, M.~Zhou, and L.~Sun, ``{Sum-Rate
  Maximization for IRS-Assisted UAV OFDMA Communication Systems},'' \emph{IEEE
  Transactions on Wireless Communications}, vol.~20, no.~4, pp. 2530--2550,
  Dec. 2021.

\bibitem{UAVIRS4}
S.~Fang, G.~Chen, and Y.~Li, ``{Joint Optimization for Secure Intelligent
  Reflecting Surface Assisted UAV Networks},'' \emph{IEEE Wireless
  Communications Letters}, vol.~10, no.~2, pp. 276--280, Sep. 2021.

\bibitem{UAVonIRS1}
T.~Shafique, H.~Tabassum, and E.~Hossain, ``{Optimization of Wireless Relaying
  With Flexible UAV-Borne Reflecting Surfaces},'' \emph{IEEE Transactions on
  Communications}, vol.~69, no.~1, pp. 309--325, Oct. 2021.

\bibitem{UAVgroundIRS1}
Y.~Cai, Z.~Wei, S.~Hu, D.~W.~K. Ng, and J.~Yuan, ``{Resource Allocation for
  Power-Efficient IRS-Assisted UAV Communications},'' \emph{2020 IEEE
  International Conference on Communications Workshops (ICC Workshops)}, pp.
  1--7, Jun. 2020.

\bibitem{UAVgroundIRS2}
L.~Ge, P.~Dong, H.~Zhang, J.~Wang, and X.~You, ``{Joint Beamforming and
  Trajectory Optimization for Intelligent Reflecting Surfaces-Assisted UAV
  Communications},'' \emph{IEEE Access}, vol.~8, pp. 78\,702--78\,712, Apr.
  2020.

\bibitem{UAVProbLOS}
A.~Al-Hourani, S.~Kandeepan, and S.~Lardner, ``{Optimal LAP Altitude for
  Maximum Coverage},'' \emph{IEEE Wireless Communications Letters}, vol.~3,
  no.~6, pp. 569--572, Jul. 2014.

\bibitem{CVX1}
M.~Grant and S.~Boyd, ``{CVX}: Matlab software for disciplined convex
  programming, version 2.1,'' \url{http://cvxr.com/cvx}, Mar. 2014.

\end{thebibliography}

\end{document}